\documentclass[12pt]{article}
\usepackage{dsfont}
\usepackage{graphicx}
\usepackage{amsmath}
\usepackage{amssymb}
\newtheorem{theorem}{Theorem}[section]
\newtheorem{lemma}{Lemma}[section]

\newtheorem{proposition}{Proposition}[section]
\newtheorem{definition}{Definition}[section]

\usepackage{arydshln}
\newenvironment{proof}[1][Proof]{\noindent\textbf{#1.} }{\ }
\begin{document}
\title{Formulae for entanglement in a linear coherent feedback network of multiple nondegenerate optical parametric amplifiers: the infinite bandwidth case}
\author{Zhan Shi and Hendra I. Nurdin 
\thanks{
Z. Shi and H. I. Nurdin are with School of Electrical Engineering and 
Telecommunications,  UNSW Australia,  
Sydney NSW 2052, Australia (e-mail: zhan.shi@student.unsw.edu.au,  h.nurdin@unsw.edu.au).} 
}
\maketitle


\begin{abstract}
This paper presents formulae for Einstein-Podolsky-Rosen (EPR) entanglement generated from $N$ nondegenerate optical parametric amplifiers (NOPAs) interconnected in a linear coherent feedback (CFB) chain in the idealized  lossless scenario and infinite bandwidth limit. The lossless scenario sets the ultimate EPR entanglement (two-mode squeezing) that can be achieved by this linear chain of NOPAs  while the infinite bandwidth limit simplifies the analysis but gives an accurate approximation to the EPR entanglement at low frequencies of interest. Two adjustable phase shifts are placed at the outputs of the system to achieve the best EPR entanglement by selecting appropriate quadratures of the output fields.
\end{abstract}

\section{Introduction}
\label{sec:intro}
Entanglement is  a critical resource to realize applications of quantum information processing, such as quantum teleportation, quantum key distribution and entanglement swapping \cite{Yonezawa2004, Jouguet2013, Briegel1998}. 
Depending on whether the state of a system is described by discrete or continuous variables, entanglement of the system comes into two categories: discrete-variable entanglement and continuous-variable entanglement \cite{Horodecki2009, Bowen2004, Weedbrook2012}. 
This paper is concerned with Einstein-Podolsky-Rosen (EPR) entanglement which is continuous-variable entanglement generated between two continuous-mode Gaussian fields. EPR entanglement can be  efficiently prepared by a nondegenerate optical parametric amplifier (NOPA) \cite{Ou1992}.

The input/output block representation of a NOPA ($G_i$) is shown in Fig.~\ref{fig:single-NOPA}. The main component of the NOPA is a  cavity consisting of a nonlinear $\chi^{(2)}$ (second-order susceptibility) crystal with two frequency degenerate but polarization non-degenerate vacuum modes.  Exploiting a coherent pump beam (which can be regarded as an undepleted classical light) to shine the crystal, the pump beam interacts with the two vacuum modes and is split into a pair of outgoing fields, via a spontaneous parametric down-conversion process. Moreover, the outgoing fields are in the two-mode squeezed state, that is, they have correlated amplitude quadratures and anti-correlated phase quadratures \cite{Ou1992}. The degree of the two-mode squeezing is quantified by both spectra variances of the difference of the amplitude quadratures and the sum of the phase quadratures. The two beams are EPR entangled at frequency $\omega$ if the variances at $\omega$ are squeezed below the quantum shot-noise limit (SNL) \cite{Vitali2006}. EPR entanglement generated by a NOPA exists in the low frequency band, but not the whole band \cite{He2007}.

\begin{figure}[htbp]
\begin{center}
\includegraphics[scale=0.5]{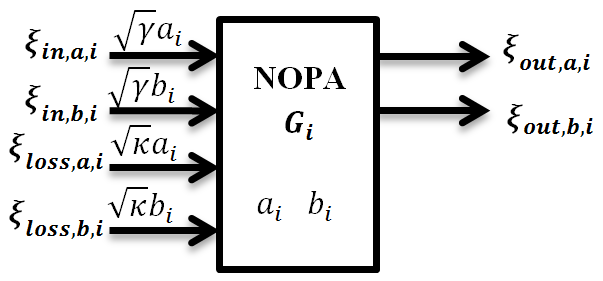}
\caption{Input/output block representation of a NOPA.}\label{fig:single-NOPA}
\end{center}
\end{figure}

EPR entanglement can be enhanced by connecting NOPAs in a coherent feedback (CFB) scheme. Our previous work \cite{SN2015qip} shows that compared with a single NOPA system and a two-cascaded NOPA network, a dual-NOPA CFB system generates stronger EPR entanglement in both lossless and lossy cases.  In our recent work \cite{SN2015qic}, we extend the dual-NOPA CFB system to an $N$-NOPA  CFB chain as shown in Fig.~\ref{fig:N-NOPA-cfb}.  The NOPAs are identical and two adjustable phase shifters with phase shifts $\theta_a$ and $\theta_b$ are placed at fields $\xi_{out,a,N}$ and $\xi_{out,b,1}$ separately in order to obtain the best two-mode squeezing between fields $\xi_{out,a}$ and $\xi_{out,b}$. 
The EPR entanglement of the system can be shared by two distant communicating parties, say Alice and Bob. One outgoing field $\xi_{out,a}$ is sent to Alice, and the other one $\xi_{out,b}$ is sent to Bob. The NOPAs in the system can be either centralised in between Alice and Bob, or evenly deployed between the two parties as discussed in \cite{SN2015qic}. 
The work \cite{SN2015qic} gives stability conditions of the system in both lossless and lossy cases. Moreover, we investigated the EPR entanglement when the system is lossless with $N$ up to six, and found the values of $\theta_a$ and $\theta_b$ at 
which the two-mode squeezing is optimal. Via numerically analysis,  we qualitatively showed that a system employing more NOPAs achieves  a higher squeezing level,  under the same total pump power. However,  the study of EPR entanglement in the work \cite{SN2015qic} is limited to up to six NOPAs and the degree of two-mode squeezing is not formulated quantitatively. 
\begin{figure*}[t!]
\begin{center}
\includegraphics[scale=0.4]{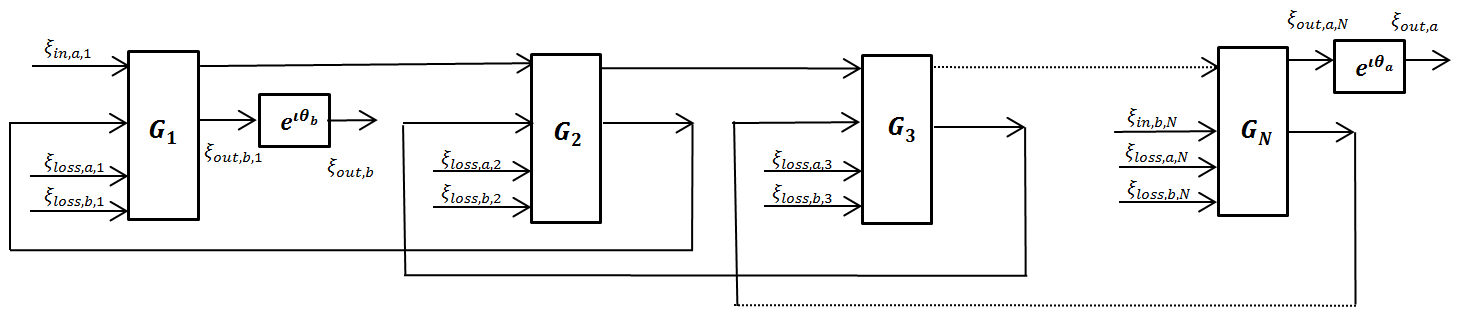}
\caption{An $N$-NOPA CFB system with  two adjustable phase shifters at the outputs.}\label{fig:N-NOPA-cfb}
\end{center}
\end{figure*}

In this paper, we derive formulae for  EPR entanglement of the $N$-NOPA CFB system ($N\geq 2$) in the ideal case where losses are neglected. We do this for two reasons. First, the lossless scenario is important  as it sets the ultimate limit of EPR entanglement (two-mode squeezing) that can be achieved by the system. Secondly, the presence of losses makes the derivation of general formulae very complicated. Thus, we leave the general case with the presence of losses for future work. To further simplify the analysis, we consider an idealized limit where the system becomes static, that is, the NOPAs are approximated as static devices with infinite bandwidth \cite{Gough2009}. Nonetheless, the infinite bandwidth setting gives an accurate approximation to the two-mode squeezing in the low frequency region. In addition, we describe the CFB system as a special case of a more general $N$-NOPA feedback network containing a linear static passive subsystem, as shown in Fig.~\ref{fig:system}. Here, the CFB connection  of Fig.~\ref{fig:N-NOPA-cfb} is recovered with a special choice of the matrix $\tilde  S_N$ in Fig.~\ref{fig:system}. In general, the matrix $\tilde S_N$ represents an arbitrary static passive network consisting of   static passive optical devices such as beams splitters, phase shifters and mirrors; see, e.g., \cite{SN2015acc}. By analyzing the general system, we obtain a transfer matrix relating ingoing and outgoing fields.  Afterwards, we aim to find the values of $\theta_a$ and $\theta_b$ at which the two-mode squeezing of the general system corresponding to the $N$-NOPA CFB system is optimal. Furthermore, we formulate the optimal two-mode squeezing. Future research building on this work may  consider exploiting the general $N$-NOPA system to analyze other linear quantum optical systems in which the $N$-NOPAs  are interconnected or connected with passive components.

The rest of the paper is structured as follows.
Section~\ref{sec:prelim} introduces descriptions of linear quantum systems, EPR entanglement between two continuous-mode Gaussian fields, and linear transformations implemented by a NOPA in the infinite bandwidth limit.  
Section~\ref{sec:system-model} describes the $N$-NOPA general system as shown in Fig.~\ref{fig:system}. We analyze the stability of the general system in the finite bandwidth case and  formulate the transfer matrix of the system in the infinite bandwidth limit. 
Section~\ref{sec:chain} analyzes the general system corresponding to a lossless $N$-NOPA CFB system. We present the values of $\theta_a$ and $\theta_b$ at which the system achieves the best two-mode squeezing and give the formulae of the optimal squeezing.
Finally, we draw a conclusion in Section~\ref{sec:conclusion}. 
\begin{figure}[htbp]
\begin{center}
\includegraphics[scale=0.9]{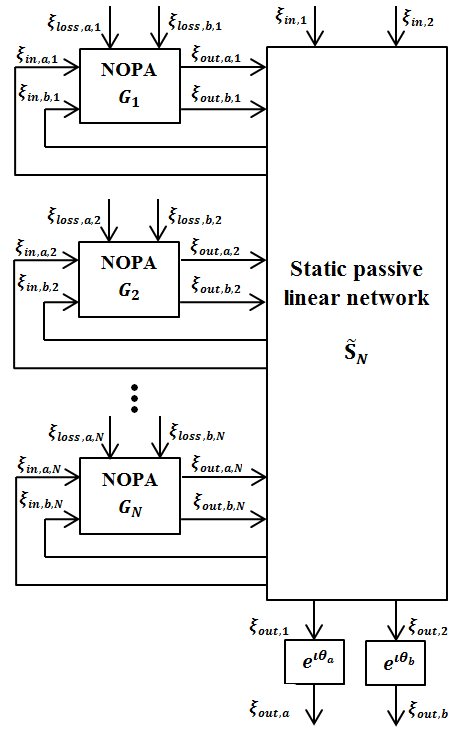}
\caption{A system consisting of $N$ NOPAs and a static passive linear subsystem with two inputs and two outputs.}\label{fig:system}
\end{center}
\end{figure}

\section{Preliminaries}
\label{sec:prelim}
The notations used in this paper are as follows: $\imath=\sqrt{-1}$. The conjugate of a matrix is denoted by $\cdot^\#$, $\cdot^T$ denotes the transpose of a matrix of numbers or operators and $\cdot^*$ denotes (i) the complex conjugate of a number, (ii) the conjugate transpose of a matrix, as well as (iii) the adjoint of an operator. $O_{m\times n}$ is an $m$ by $n$ zero matrix and $I_n$ is an $n$ by $n$ identity matrix (we simply write $O$ and $I$, if the dimensions of the matrix can be inferred from context). Trace operator is written as $\operatorname{Tr[\cdot]}$ and tensor product is $\otimes$.  $\delta(t)$ denotes the Dirac delta function. 
\subsection{Linear quantum systems}
\label{sec:linear_sys}
Here we consider an open linear quantum system without a scattering process. The linear system contains $n$-bosonic modes $a_j(t)~(j=1,\ldots, n)$ satisfying the commutation relations $[a_i(t), a_j(t)^*]=\delta_{ij}$, $m$-incoming boson fields $\xi_{in,i}(t)~(i=1,\ldots, m)$ in the vacuum state, which obey the commutation relations $[\xi_{in,j}(t), \xi_{in,j}(s)^*]=\delta(t-s)$, as well as two outgoing fields $\xi_{out,k}(t)~(k=1,2)$ which are Gaussian continuous-mode fields. A continuous-mode field means that the field contains a continuum of modes in a continuous range of frequencies.  
Note that a system may have more than two outputs. 
However, as we are only interested in the entanglement generated by a certain pair of outgoing fields, in this work we will only be interested in a particular pair of output fields, labelled ${out,1}$ and ${out,2}$ in the following.
The time-varying interaction Hamiltonian between the system and its environment is $H_{\rm int}(t) = \imath (\xi(t)^*L - L^* \xi(t))$, where $\xi(t)=[\xi_{in,1}(t),\ldots \xi_{in,m}(t)]^T $, $L=[L_1,\ldots, L_l]^T$ and $L_j  (j=1, 2, \cdots, l)$ is the $j$-th system coupling operator. In the Heisenberg picture, time evolutions of a mode $a_j$ and an outgoing field operator $\xi_{out,i}$ are \cite{bGardiner2004,Nurdin2009}:
\begin{eqnarray}
a_j(t)&=&U(t)^* a_j U(t),\nonumber\\
\xi_{out,i}(t)&=&U(t)^*\xi_{in,i}(t)U(t),
\end{eqnarray}
where $U(t)={\rm exp}^{\hspace{-0.5cm}\longrightarrow}~(-i\int_0^t H_{\rm int}(s)ds)$ is a unitary process obeying the quantum white noise Schr\"{o}dinger equation $\dot{U}(t)=-\imath H_{\rm int}(t)U(t)$. 
However, this is not an ordinary Schr\"{o}dinger equation as the interaction Hamiltonian $H_{int}(t)$  is a time-varying observable involving the singular quantum white noise processes $\xi(t)$.  This quantum white noise equation has to be interpreted correctly within the framework of quantum stochastic calculus, for details see \cite{bGardiner2004, inbBelavkin2008, bWiseman2010, Gough2003}.
Employing quantum stochastic calculus,  the dynamics of a linear quantum system is described by a quantum Langevin equation and can be written in the following form
\begin{eqnarray}
    \dot{z}(t)&=&Az(t)+B\xi(t), \label{eq:dynamics} \\
     \xi_{out}(t)&=&Cz(t)+D\xi(t). \label{eq:output}
\end{eqnarray}
where
\begin{eqnarray}
    z&=&(a_1^q, a_1^p, \ldots, a_n^q, a_n^p)^T, \nonumber\\
    \xi & =&(\xi_{1}^q, \xi_{1}^p, \ldots, \xi_{m}^q, \xi_{m}^p)^T, \nonumber\\
    \xi_{out}&=&(\xi_{out,1}^q, \xi_{out,1}^p, \xi_{out,2}^q, \xi_{out,2}^p)^T, \label{eq:vector}
\end{eqnarray}
with {\it quadratures} \cite{inbBelavkin2008,bWiseman2010}
\begin{eqnarray}
a_j^q &=& a_j+a_j^*, \quad a_j^p = (a_j-a_j^*)/i, \nonumber \\
\xi_j^q &=& \xi_j+\xi_j^*, \quad \xi_j^p = (\xi_j-\xi_j^*)/i. \label{eq: quadratures}
\end{eqnarray}
  
\subsection{EPR entanglement between two continuous-mode fields}
\label{sec:entanglement}
We keep in mind that in this paper, we investigate EPR entanglement between two continuous-mode Gaussian fields rather than entanglement between two single-mode Gaussian fields. In the latter case, the degree of entanglement can be assessed via the logarithmic negativity as an entanglement measure, see, e.g., \cite{Laurat2005}. However, this measure is not directly applicable  to continuous-mode fields. Instead, the EPR entanglement of two freely propagating fields containing a continuum of modes, say $\xi_{out,1}$ and $\xi_{out,2}$, can be evaluated in the frequency domain by the two-mode squeezing spectra $V_+(\imath\omega)$  and $V_-(\imath\omega)$ \cite{Ou1992,Vitali2006}, that will be defined below.

The Fourier transform of $f(t)$ is defined as $F\left(\imath\omega\right)=\frac{1}{\sqrt{2\pi}}\int_{-\infty}^{\infty} f\left(t\right)e^{-\imath\omega t} dt$. Similarly, we have the Fourier transforms of $\xi_{out,1}(t)$, $\xi_{out,2}(t)$, $z(t)$ and $\xi(t)$ in (\ref{eq:dynamics}) and (\ref{eq:output}) as  $\tilde \Xi_{out,1}\left(\imath\omega\right)$, $\tilde \Xi_{out,2}\left(\imath\omega\right)$, $Z(\imath \omega)$  and  $\Xi(\imath \omega)$, respectively. 
Applying (\ref{eq:dynamics}), (\ref{eq:output}), we have
\begin{eqnarray}
\tilde \Xi_{out,1}^q(\imath \omega)+\tilde \Xi_{out,2}^q(\imath \omega)  
&=&\int_{-\infty}^{\infty} \xi_{out,1}^q(t)e^{-\imath\omega t} dt+\int_{-\infty}^{\infty} \xi_{out,2}^q(t)e^{-\imath\omega t}  dt  \nonumber \\
&=&  [1\ 0\ 1\ 0] \left(C Z\left(\imath\omega\right)+ D\Xi\left(\imath\omega\right)\right) , \nonumber\\
\tilde \Xi_{out,1}^p(\imath \omega)-\tilde \Xi_{out,2}^p(\imath \omega) 
&= & \int_{-\infty}^{\infty} \xi_{out,1}^p(t)e^{-\imath\omega t} dt-\int_{-\infty}^{\infty} \xi_{out,2}^p(t)e^{-\imath\omega t}  dt  \nonumber \\
&=& [0\ 1\ 0\ {-}1]\left(C Z\left(\imath\omega\right)+ D\Xi\left(\imath\omega\right)\right).
\end{eqnarray}

The two-mode squeezing spectra $V_{+}(\imath \omega)$ and $V_{-}(\imath \omega)$ are real functions defined via the identities
\begin{eqnarray}
\langle (\tilde \Xi_{out,1}^q(\imath \omega)+\tilde \Xi_{out,2}^q(\imath \omega))^* (\tilde \Xi_{out,1}^q(\imath \omega')+\tilde \Xi_{out,2}^q(\imath \omega')) \rangle = V_+(\imath \omega)\delta(\omega-\omega'), \nonumber \\
\langle (\tilde \Xi_{out,1}^p(\imath \omega)+\tilde \Xi_{out,2}^p(\imath \omega))^* (\tilde \Xi_{out,1}^p(\imath \omega')-\tilde \Xi_{out,2}^p(\imath \omega')) \rangle =  V_-(\imath \omega) \delta(\omega-\omega'),
\end{eqnarray}
where $\langle \cdot \rangle$ denotes quantum expectation. As described in \cite{Gough2010,Nurdin2012}, $V_+(\imath \omega)$ and $V_-(\imath \omega)$ are easily calculated by,
\begin{eqnarray}
V_+(\imath\omega)&=& {\rm Tr}\left[H^Q(\imath\omega)^* H^Q(\imath\omega)\right], \label{eq:V_+}\\
V_-(\imath\omega)&=& {\rm Tr}\left[H^P(\imath\omega)^* H^P(\imath\omega)\right], \label{eq:V_-}
\end{eqnarray}
where $H^Q=[1\ 0\ 1\ 0]H$, $H^P=[0\ 1\ 0\ {-}1]H$ and $H$ is the transfer function
\begin{eqnarray}
H(\imath\omega)=C\left(\imath\omega I-A \right)^{-1}B+D. \label{eq:transfer-function}
\end{eqnarray}
The fields $\xi_{out,1}$ and  $\xi_{out,2}$ are said to be EPR-entangled at the frequency $\omega$ rad/s is \cite{Vitali2006},
\begin{eqnarray}
V(\imath\omega) = V_+(\imath\omega)+V_-(\imath\omega)< 4, \label{eq:entanglement-criterion}
\end{eqnarray}
which indicates that the two-mode squeezing level is below the quantum shot-noise limit at $\omega$ rad/s \cite{Vitali2006}.

A perfect Einstein-Podolski-Rosen state is represented by an infinite bandwidth two-mode squeezing, that is $V(\imath\omega) = V_{\pm}(\imath\omega)= 0$ for all $\omega$. Of course, such an ideal EPR correlation cannot be achieved in reality as it would require an infinite amount of energy to produce. Thus, we aim to optimize EPR entanglement by making $V(\imath \omega)$ as small as possible over a wide frequency range \cite{Vitali2006}. 
Note that (\ref{eq:entanglement-criterion}) is a sufficient condition for EPR entanglement, with the two beams squeezed in amplitude and phase quadratures. However, in general, they may be squeezed in other quadratures. Hence, we give the following definition of EPR entanglement.
Let $\xi^{\psi_1}_{out,1}=e^{\imath \psi_1}\xi_{out,1}$, $\xi^{\psi_2}_{out,2}=e^{\imath \psi_2}\xi_{out,2}$ with $\psi_1, \psi_2 \in(-\pi,\pi]$ and denote the corresponding two-mode squeezing spectra between  $\xi^{\psi_1}_{out,1}$ and $\xi^{\psi_2}_{out,2}$ as $V^{\psi_1, \psi_2}_\pm(\imath\omega,\psi_1, \psi_2)$.
\begin{definition}
The fields $\xi_{out,1}$ and  $\xi_{out,2}$ are EPR entangled at the frequency $\omega$ rad/s if $\exists ~\psi_1, \psi_2 \in(-\pi,\pi]$ such that
\begin{eqnarray}
 V^{\psi_1, \psi_2}_+(\imath\omega,\psi_1, \psi_2)+V^{\psi_1, \psi_2}_-(\imath\omega,\psi_1, \psi_2)< 4. \label{eq:entanglement-criterion-2}
\end{eqnarray}
EPR entanglement is said to vanish at $\omega$ if there are no values of $\psi_1$ and $\psi_2$ satisfying the above criterion.
Unless otherwise specified, throughout the paper, EPR entanglement refers to the case with $\psi_1=\psi_2=0$.
\end{definition}

\subsection{The nondegenerate optical parametric amplifier (NOPA)}
\label{sec:NOPA}
A NOPA ($G_i$) is an open linear quantum system containing a cavity with a pair of orthogonally polarized bosonic modes $a_i$ and $b_i$ which satisfy  $[a_i, a_j^*]=\delta_{ij}$, $[b_i, b_j^*]=\delta_{ij}$, $[a_i, b_j^*]=0$ and $[a_i, b_j]=0$. By assuming a strong undepleted coherent pump beam onto the $\chi^{(2)}$ nonlinear crystal inside the cavity, the pump can be treated as a classical field (hence, quantum vacuum fluctuations are ignored) and the interaction of the modes $a_i$ and $b_i$ with the pump is modelled by the two-mode squeezing Hamiltonian $H= \frac{\imath}{2} \epsilon\left( a_i^* b_i^*- a_ib_i\right)$, where $\epsilon$ is a real coefficient relating to the effective amplitude of the pump beam,  see \cite{bGardiner2004, bCarmichael2008, bBachor2009}.

As shown in Fig.~\ref{fig:single-NOPA}, interactions between the NOPA and its environment are denoted by coupling operators as follows. Modes $a_i$ and $b_i$ are coupled to ingoing fields $\xi_{in,a,i}$ and $\xi_{in,b,i}$ via coupling operators $L_1=\sqrt{\gamma}a_i$ and $L_2=\sqrt{\gamma}b_i$, respectively. Unwanted amplification losses $\xi_{loss,a,i}$ and $\xi_{loss,b,i}$ impact the NOPA through operators $L_3=\sqrt{\kappa}a_i$ and $L_4=\sqrt{\kappa}b_i$, respectively. The constants $\gamma$ and $\kappa$ are damping rates of the outcoupling mirrors (from which the output fields emerge from the NOPA), and of the loss channels, respectively. 
From Section \ref{sec:linear_sys}, we have that the dynamics of the NOPA is given by\cite{Ou1992, Nurdin2009}
\begin{eqnarray}
\dot{a_i}\left(t\right)&=&-\left(\frac{\gamma+\kappa}{2}\right)a_i\left(t\right)+\frac{\epsilon}{2}b_i^*\left(t\right)-\sqrt{\gamma}\xi_{in,a,i}\left(t\right)-\sqrt{\kappa}\xi_{loss,a,i}\left(t\right),\nonumber \\
\dot{b_i}\left(t\right)&=&-\left(\frac{\gamma+\kappa}{2}\right)b_i\left(t\right)+\frac{\epsilon}{2}a_i^*\left(t\right)-\sqrt{\gamma}\xi_{in,b,i}\left(t\right)-\sqrt{\kappa}\xi_{loss,b,i}\left(t\right), \label{eq:NOPA-dynamics1}
\end{eqnarray}
following the boundary conditions \cite{Ou1992, bGardiner2004}, and we have outputs
\begin{eqnarray}
\xi_{out,a,i}\left(t\right)&=&\sqrt{\gamma}a_i\left(t\right)+\xi_{in,a,i}\left(t\right),\nonumber \\
\xi_{out,b,i}\left(t\right)&=&\sqrt{\gamma}b_i\left(t\right)+\xi_{in,b,i}\left(t\right). \label{eq:NOPA-dynamics2}
\end{eqnarray}

Define the following quadrature vectors of the NOPA, 
\begin{eqnarray}
z &=& [a^q_i, a^p_i, b^q_i, b^p_i]^T, \nonumber\\
\xi_{out}&=&[\xi^q_{out,a,i},\xi^p_{out,a,i},\xi^q_{out,b,i},\xi^p_{out,b,i}]^T,\nonumber \\
\xi &=&[\xi^q_{in,a,i},\xi^p_{in,a,i},\xi^q_{in,b,i},\xi^p_{in,b,i}, \xi^q_{loss,a,i},\xi^p_{loss,a,i},\xi^q_{loss,b,i},\xi^p_{loss,b,i}]^T.
\end{eqnarray}
From (\ref{eq:dynamics}), (\ref{eq:output}) and (\ref{eq:transfer-function}), the transfer function of the NOPA is
\begin{eqnarray}
H_{N}=
\left[\begin{array}{cccccccc}
h_1 & 0 & h_2 & 0 & h_3 & 0 & h_4 & 0 \\
0 & h_1 & 0 & -h_2 & 0 & h_3 & 0 & -h_4 \\
h_2 & 0 & h_1 & 0 & h_4 & 0 & h_3 & 0 \\
0 & -h_2 & 0 & h_1 & 0 & -h_4 & 0 & h_3
\end{array} \right],\label{eq:NOPAtf}
\end{eqnarray}
where $h_j~ (j=1,2,3,4)$ are functions of the frequency $\omega$,
\begin{eqnarray}
h_1(\imath \omega) =\frac{\epsilon^2+\gamma^2-(\kappa+2\imath\omega)^2}{\epsilon^2-(\gamma+\kappa+2\imath\omega)^2}, ~~h_2(\imath \omega) =\frac{2\epsilon\gamma}{\epsilon^2-(\gamma+\kappa+2\imath\omega)^2},\nonumber\\
h_3 (\imath \omega) =\frac{2\sqrt{\gamma\kappa}(\gamma+\kappa+2\imath\omega)}{\epsilon^2-(\gamma+\kappa+2\imath\omega)^2}, ~~h_4 (\imath \omega) =\frac{2\epsilon\sqrt{\gamma\kappa}}{\epsilon^2-(\gamma+\kappa+2\imath\omega)^2}. \label{eq:h_freq_dependent}
\end{eqnarray}

As reported in \cite{Nurdin2009,Iida2012}, parameters of the NOPA are set as follows. We set the reference value for the transmissivity rate of the mirrors $\gamma_{r}=7.2\times10^7$ Hz. The pump amplitude $\epsilon$ and damping rate $\gamma$ are adjustable as $\epsilon=x\gamma_r$ and $\gamma=\frac{\gamma_r}{y}$ respectively, where the variables $x$ and $y$ satisfy $0<x,y \leq 1$. We set $\kappa=K\epsilon$ with $K=\frac{3\times 10^6}{\sqrt{2}\times 0.6 \times \gamma_r}$ based on the assumption that the value of $\kappa$ is proportional to the absolute value of $\epsilon$ and $\kappa= \frac{3 \times 10^6}{\sqrt{2}}$ when $\epsilon=0.6 \gamma_r$. In the infinite bandwidth case where we take the limit $\gamma_r\rightarrow\infty$ while keeping $\epsilon$ and $\gamma$ at a fixed ratio $\frac{\epsilon}{\gamma}=\frac{x}{y}$, the transfer function of the NOPA in (\ref{eq:NOPAtf}) becomes a constant matrix with elements
\begin{eqnarray}
h_1&=&\frac{(1-K^2)(xy)^2+1}{(xy)^2-(1+Kxy)^2},~~h_2=\frac{2xy}{(xy)^2-(1+Kxy)^2},\nonumber\\
h_3&=&\frac{2\sqrt{Kxy}(1+Kxy)}{(xy)^2-(1+Kxy)^2},~~h_4=\frac{2xy\sqrt{Kxy}}{(xy)^2-(1+Kxy)^2}. \label{eq:h_coeffi_static}
\end{eqnarray}

It can be seen that for $\omega \ll \epsilon, \gamma, \kappa$, the constant scalar values of $h_1$ to $h_4$ given by (\ref{eq:h_coeffi_static}) in the infinite bandwidth limit approximates the frequency dependent values given in (\ref{eq:h_freq_dependent}) when the bandwidth is finite. Such an approximation is quite accurate for $\omega$ sufficiently small, away from $\epsilon, \gamma, \kappa$ (with no error at $\omega=0$). Since in practice the EPR entanglement will be in the low frequency region, entanglement in the idealised infinite bandwidth scenario provides a good approximation for the entanglement that can be expected in the finite bandwidth case.
\section{$N$-NOPA general system}
\label{sec:system-model}
We consider a linear quantum system consisting of $N$ ($N\geq 2$) NOPAs and a linear static passive network as shown in Fig.~\ref{fig:system}. The static passive linear network has $2(N+1)$ ingoing fields, among which two inputs $\xi_{in,1}$ and $\xi_{in,2}$ are in the vacuum state and the others are the outgoing fields of the NOPAs. Similarly, the network has $2(N+1)$ outgoing fields, $2N$ of which are inputs of the NOPAs. The other outputs $\xi_{out,1}$ and $\xi_{out,2}$ are linked to two adjustable phase shifters with phase shifts $\theta_a$ and $\theta_b$, respectively. The EPR entanglement of interest is  between fields  $\xi_{out,a}$ and $\xi_{out,b}$.

The transformation implemented by the static subsystem is denoted by a  complex unitary matrix  $\tilde{S}_N \in \mathbb{C}^{ 2(N+1)}$ such that
\begin{eqnarray}
\left[ \begin{array}{c}
\xi_{out,1}\\ \xi_{out,2}\\ \xi_{in,a,1} \\ \xi_{in,b,1} \\ \vdots \\ \xi_{in,a,N} \\ \xi_{in,b,N} \end{array} \right] 
= \tilde{S}_N \left[ \begin{array}{c}
\xi_{in,1}\\ \xi_{in,2}\\ \xi_{out,a,1} \\ \xi_{out,b,1} \\ \vdots \\ \xi_{out,a,N} \\ \xi_{out,b,N} \end{array} \right], \label{eq:StN}
\end{eqnarray}
and
\begin{equation}
\tilde{S}_N^* \tilde{S}_N = \tilde{S}_N\tilde{S}_N^*= I_{2(N+1)}.
 \label{eq:complex-unitary}
\end{equation} 
In particular, the matrix $\tilde{S}_N$ corresponding to the $N$-NOPA CFB system as shown in Fig.~\ref{fig:N-NOPA-cfb} is denoted as
\begin{eqnarray}
\tilde{S}_{N} ^{cfb}= \left[\begin{array}{cc} O_{2N \times 2} & I_{2N} \otimes M_1\\ M_1 & O_{2 \times 2N}\end{array}\right]
                     +\left[\begin{array}{cc} O_{2 \times 2N} & M_2  \\ I_{2N} \otimes M_2&O_{2N \times 2} \end{array}\right],
\label{eq:ScfbN} \end{eqnarray}
where
\begin{eqnarray}
M_1=\left[\begin{array}{cc} 0&0\\0&1\end{array}\right],~~~M_2=\left[\begin{array}{cc} 1&0\\0&0\end{array}\right]. 
\label{eq:jdju} \end{eqnarray}

Now we define the following quadrature vectors
\begin{eqnarray}
z&=&[a^q_1, a^p_1, b^q_1, b^p_1,\cdots, a^q_N, a^p_N, b^q_N, b^p_N]^T,\nonumber \\
\xi_{in}&=&[\xi^q_{in,a,1},\xi^p_{in,a,1},\xi^q_{in,b,1},\xi^p_{in,b,1}, \cdots, \xi^q_{in,a,N},\xi^p_{in,a,N},\xi^q_{in,b,N},\xi^p_{in,b,N}]^T, \nonumber\\
\xi_{loss}&=&[\xi^q_{loss,a,1},\xi^p_{loss,a,1},\xi^q_{loss,b,1},\xi^p_{loss,b,1}, \cdots,\xi^q_{loss,a,N},\xi^p_{loss,a,N},\xi^q_{loss,b,N},\xi^p_{loss,b,N}]^T, \nonumber\\
\xi_{out}&=&[\xi^q_{out,a,1},\xi^p_{out,a,1},\xi^q_{out,b,1},\xi^p_{out,b,1}, \cdots,\xi^q_{out,a,N},\xi^p_{out,a,N},\xi^q_{out,b,N},\xi^p_{out,b,N}]^T, \nonumber\\
\xi^{(i)}&=&[\xi^q_{in,1},\xi^p_{in,1},\xi^q_{in,2},\xi^p_{in,2}]^T, \nonumber\\
\xi^{(o)}&=&[\xi^q_{out,1},\xi^p_{out,1},\xi^q_{out,2},\xi^p_{out,2}]^T, \nonumber\\
\xi &=&[{\xi^{(i)}}^T,\xi_{loss}^T]^T. \label{eq:sys_quadratures}
\end{eqnarray}
Following (\ref{eq: quadratures}), the quadrature form of $\tilde{S}_N$ is defined as
\begin{align}
S_N=\frac{1}{2}\tilde{K}_N\tilde{S}_N\tilde{K}_N^* + \frac{1}{2}\tilde{K}_N^\# \tilde{S}_N^\#\tilde{K}_N^T,  \label{eq:relations-real-complex-matrix}
\end{align}
where $\tilde{K}_N= I_{2(N+1)} \otimes \left[ \begin{array}{cc} 1 & -\imath  \end{array} \right]^T$.
Note that $S_N$ is a real unitary symplectic matrix. That is,
\begin{eqnarray}
S_N^TS_N=S_NS_N^T=I_{4(N+1)}, ~~~~S_N^T \mathbb{J}_{2(N+1)} S_N =\mathbb{J}_{2(N+1)} \label{eq:unitary_symplectic}
\end{eqnarray}
where $\mathbb{J}_{2(N+1)}=I_{2(N+1)}\otimes \left[ \begin{array}{cc} 0 & 1 \\ -1 & 0 \end{array} \right] $. 

\subsection{Stability} \label{sec:stability}
Though we analyze the $N$-NOPA  general system in the infinite bandwidth case where stability is not an issue as the system is assumed static, we have to guarantee that the system in practice (the finite bandwidth case) is stable, that is, the average total number of photons in the cavity modes does not continuously grow.  To investigate stability, we write the dynamics the $N$-NOPA  general system in the finite bandwidth case in the form
\begin{eqnarray}
\dot{z}(t)&=&A_N z(t)+B_N \xi(t).
 \label{eq:N-NOPA-dynamics_standard}
\end{eqnarray}
Then, the system is sable if and only if $A_N$ is Hurwitz (all the eigenvalues of the matrix have real negative parts).

For convenience, we write $S_N$ in the form 
\begin{eqnarray}
S_N=\left[\begin{array}{cc} S_{N,11} &S_{N,12}\\S_{N,21}&S_{N,22} \end{array}\right],\label{eq:S_sub_matrice_form}
\end{eqnarray}
where $S_{N,11} \in \mathbb{R}^{4 \times 4} $, $S_{N,21} \in \mathbb{R}^{4N \times 4} $,$S_{N,12} \in \mathbb{R}^{4 \times 4N} $
and $S_{N,22} \in \mathbb{R}^{4N \times 4N} $
According to (\ref{eq:NOPA-dynamics1}), (\ref{eq:NOPA-dynamics2}) and (\ref{eq:sys_quadratures}), we have
\begin{eqnarray}
\dot{z}(t)&=&\left(I_N \otimes A_1\right)z(t)-\sqrt{\gamma}\xi_{in}(t)-\sqrt{\kappa}\xi_{loss}(t), \nonumber\\
\xi_{out}&=&\sqrt{\gamma}z(t)+\xi_{in}(t),
 \label{eq:N-NOPA-dynamics}
\end{eqnarray}
where
\begin{eqnarray}
A_1=\left[\begin{array}{cccc} -\frac{\gamma+\kappa}{2} &0&\frac{\epsilon}{2}&0\\
0&-\frac{\gamma+\kappa}{2} &0&-\frac{\epsilon}{2}\\
\frac{\epsilon}{2}&0&-\frac{\gamma+\kappa}{2} &0\\
0&-\frac{\epsilon}{2}&0&-\frac{\gamma+\kappa}{2}\\ \end{array}\right]. \label{eq:A1}
\end{eqnarray}
 Substituting (\ref{eq:S_sub_matrice_form}) into (\ref{eq:StN}),  we obtain
 \begin{eqnarray}
\xi_{in}=S_{N,21}\xi^{(i)}+S_{N,22} \xi_{out}, 
 \label{eq:xi_in}
\end{eqnarray}
where $\xi_{in}$ is as given in  (\ref{eq:sys_quadratures}). From (\ref{eq:N-NOPA-dynamics}) and (\ref{eq:xi_in}), we have
 \begin{eqnarray}
A_N=I_N \otimes A_1- \gamma (I_{4N}-S_{N,22})^{-1}  S_{N,22}. \label{eq:A_N}
\end{eqnarray}
Thus, the closed-loop $N$-NOPA general system is well-posed if $(I_{4N}-S_{N,22})$ is invertible \cite{Gough2008}. 
In addition, the system is stable if and only if the matrix $(I_N \otimes A_1- \gamma (I_{4N}-S_{N,22})^{-1} S_{N,22})$ is Hurwitz. 
In particular, a detailed sufficient and necessary stability condition of the $N$-NOPA CFB system has been developed  in \cite{SN2015qic}.
\subsection{Transformation of static $N$-NOPA system} \label{sec:transformation}
Now we investigate the static transformation $H_N$ ($\xi^{(o)}=H_N\xi$) implemented by the $N$-NOPA general system in the infinite bandwidth case. By exploiting the static transfer matrices of a NOPA and the static subsystem given by (\ref{eq:NOPAtf}), (\ref{eq:h_coeffi_static}) and (\ref{eq:StN}), we obtain
\begin{eqnarray}
H_N&=& \left(S_{N,11}+S_{N,12}\left(I_N \otimes W_{12}\right)P_N S_{N,21}\right)\left[\begin{array}{cc}I_{4}&O_{4\times 4N}\end{array}\right]   \nonumber\\
&& + S_{N,12}\left(I_N \otimes W_{34}\right)\left[\begin{array}{cc}O_{4N \times 4}&I_{4N} \end{array}\right]\nonumber\\
 &&+ \left(S_{N,12}\left(I_N \otimes W_{12}\right)P_N S_{N,22}\left(I_N \otimes W_{34}\right) \right)\left[\begin{array}{cc}O_{4N \times 4}&I_{4N} \end{array}\right] \nonumber
\end{eqnarray}
where
\begin{eqnarray}
P_N &=& \left(I_{4N}-S_{N,22}\left(I_N \otimes W_{12}\right)\right)^{-1}, \nonumber \\
W_{12}&=&\left[\begin{array}{cc}h_1I_2&h_2R\\h_2R&h_1I_2\end{array}\right],~~W_{34}=\left[\begin{array}{cc}h_3I_2&h_4R\\h_4R&h_3I_2\end{array}\right], \nonumber\\
R &=&\left[\begin{array}{cc}1&0\\0&-1\end{array} \right].\label{eq:H}
\end{eqnarray}

As shown in (\ref{eq:H}), the condition that the matrix $H_N$ is well-defined is that the matrix $\left(I_{4N}-S_{N,22}\left(I_N \otimes W_{12}\right)\right)$ is invertible.
The lemma below shows that  $H_N$ is well defined as long as the general system is stable.
\begin{lemma}\label{lemma:PI_invertible}
The matrix $\left(I_{4N}-S_{N,22}\left(I_N \otimes W_{12}\right)\right)$ is invertible if the $N$-NOPA general  system is stable.
\end{lemma}
\begin{proof}
When the system is well-posed and stable, from Section~\ref{sec:stability}, we have  
\begin{eqnarray}
\det \left(\lambda I_{4N}-(I_N \otimes A_1- \gamma (I_{4N}-S_{N,22})^{-1}  S_{N,22})\right) \neq 0, ~ \forall ~{\rm real}(\lambda) \geq 0. \nonumber
\end{eqnarray}
Thus, at $\lambda=0$,
\begin{eqnarray}
\det \left(I_N \otimes A_1- \gamma (I_{4N}-S_{N,22})^{-1}  S_{N,22}\right) \neq 0.
\end{eqnarray}
Moreover, as $\det(A_1)=\frac{1}{16}\left(\epsilon^2- \left(\gamma+\kappa \right)^2\right) \neq 0$, we have
\begin{eqnarray}
&& \det\left(I_{4N}-S_{N,22}\right) \det \left(I_N \otimes A_1- \gamma (I_{4N}-S_{N,22})^{-1}  S_{N,22}\right) \det \left(I_N \otimes A_1 \right) \nonumber\\
&=&\det\left(I_{4N}-S_{N,22}-\gamma S_{N,22} \left( I_N \otimes A_1^{-1} \right) \right) \neq 0.
\end{eqnarray}

Following (\ref{eq:H}) and (\ref{eq:S_sub_matrice_form}), we have
\begin{eqnarray}
&&\det\left(I_{4N}-S_{N,22}\left(I_N \otimes W_{12}\right)\right) \nonumber \\
&=&\det\left(I_{4N}-S_{N,22}\left(I_N \otimes \left(I_{4}+\gamma A_1^{-1}\right)\right)\right)\nonumber \\
&=&\det\left(I_{4N}-S_{N,22}-\gamma S_{N,22}\left(I_N \otimes A_1^{-1}\right)\right).
 \label{eq:Pinverse}
\end{eqnarray}
The proof is completed.
\end{proof}
\section{EPR entanglement}
\label{sec:chain}
This section analyzes EPR entanglement generated by a lossless  ($\kappa=0$) $N$-NOPA CFB system. From (\ref{eq:V_+}) and (\ref{eq:V_-}), we define the two-mode squeezing spectra $V_\pm (\imath \omega, \theta_a, \theta_b)$ as  a function of $\omega$, $\theta_a$ and $\theta_b$. In the infinite bandwidth limit, $V_\pm (\imath \omega, \theta_a, \theta_b)$ is independent of $\omega$, that is, they have the same value for all $\omega$ at fixed values of $\theta_a$ and $\theta_b$. Thus, we denote the two-mode squeezing in the infinite bandwidth limit as $V_\pm (\theta_a, \theta_b)$. 
We aim to find values of the phase shifts $\theta_a$ and $\theta_b$ at which the two-mode squeezing of the system is optimal, that is, the two-mode squeezing spectra $V_\pm (\theta_a, \theta_b)$ in the infinite bandwidth limit is minimized. After that, we formulate the degree of the optimal two-mode squeezing, denoted as $V_\pm$.  Before the analysis, we give the following definition and two useful propositions.
\begin{definition} \label{def:L2-matrix}
Let $E=\left[E_{i,j}\right]$ denote a block matrix whose sub-matrix is $E_{i,j} \in \mathbb{R}^{2\times 2}$, where $1\leq i, j \leq 2N$. The matrix $E$ is an $L^\curvearrowright _{2\times 2}$-matrix if $E$ satisfies
\begin{enumerate}
\item  \begin{eqnarray}
E_{i,j}=\left\{ \begin{array}{ll}
e_{ij}I_2 & {\rm for}~i+j~{\rm is~even}\\ 
e_{ij}R & {\rm for}~i+j~{\rm is~odd},
\end{array}\right.  \label{eq:L2matrix_property1}
\end{eqnarray}
where $e_{i,j} \in \mathbb{R}$; 
\item 
\begin{eqnarray} 
E_{i,j}=E_{2N+1-i,2N+1-j}, \label{eq:L2matrix_property2}
\end{eqnarray} 
for $1 \leq i \leq 2N$ and $N+1 \leq j \leq 2N$.
\end{enumerate}
\end{definition}
\begin{proposition} \label{propo:1}
If $E=\left[E_{i,j}\right]$ and $F=\left[F_{i,j}\right]$ $\left(1\leq i, j \leq 2N \right)$ are $L^\curvearrowright _{2\times 2}$-matrices, then $G=EF$ is also an $L^\curvearrowright _{2\times 2}$-matrix.
\end{proposition}
\begin{proof}
Let $G=\left[G_{i,j}\right]$, where $G_{i,j}\in \mathbb{R}^{2\times 2}$. We obtain that $G_{i,j}=\sum\limits_{k=1}^{2N}E_{i,k}F_{k,j}$.  
As $E$ and $F$ are $L^\curvearrowright _{2\times 2}$-matrices, from (\ref{eq:L2matrix_property1}), it is easy to check that
\begin{eqnarray}
G_{i,j}=\left\{ \begin{array}{ll}
g_{ij}I_2 & {\rm for}~i+j~{\rm is~even}\\ 
g_{ij}R & {\rm for}~i+j~{\rm is~odd},
\end{array}\right.  \label{eq:Gij}
\end{eqnarray}
where $g_{i,j}=\sum\limits_{k=1}^{2N}e_{i,k}f_{k,j}$.  
Moreover, when  $1 \leq i \leq 2N$ and $N+1 \leq j \leq 2N$,
\begin{eqnarray}
G_{2N+1-i,2N+1-j}&= &\sum\limits_{k=1}^{2N}E_{2N+1-i,k}F_{k,2N+1-j} \nonumber\\
&=&\sum\limits_{k=1}^{N}E_{2N+1-i,k}F_{k,2N+1-j}+\sum\limits_{k=N+1}^{2N}E_{i,2N+1-k}F_{k,2N+1-j}, \nonumber
\end{eqnarray}
\begin{eqnarray}
G_{i,j}&=& \sum\limits_{k=1}^{2N}E_{i,k}F_{k,j}  \nonumber\\
&=&\sum\limits_{k=1}^{N}E_{i,k}F_{2N+1-k,2N+1-j} +\sum\limits_{k=N+1}^{2N}E_{2N+1-i,2N+1-k}F_{2N+1-k,2N+1-j}. \nonumber
\end{eqnarray}
Let $k'=2N+1-k$, and then
\begin{eqnarray}
G_{i,j}= \sum\limits_{k'=N+1}^{2N}E_{i,2N+1-k'}F_{k',2N+1-j}+\sum\limits_{k'=1}^{N}E_{2N+1-i,k'}F_{k',2N+1-j}. \nonumber
\end{eqnarray}
Therefore, $G_{i,j}=G_{2N+1-i,2N+1-j}$ for $1 \leq i \leq 2N$ and $N+1 \leq j \leq 2N$.
Thus, the matrix $G$ satisfies both properties (\ref{eq:L2matrix_property1}) and (\ref{eq:L2matrix_property2}). The proof is completed.
\end{proof}

\begin{proposition} \label{propo:2}
Let $E=\left[E_{i,j}\right]$ be an invertible $L^\curvearrowright _{2\times 2}$-matrix, then $F=E^{-1}$ is also an $L^\curvearrowright _{2\times 2}$-matrix.
\end{proposition}
\begin{proof}
Suppose $F=[F_{i,j}]$ is an $L^\curvearrowright _{2\times 2}$-matrix, with  
\begin{eqnarray}
F_{i,j}=\left\{ \begin{array}{ll}
f_{ij}I_2 & {\rm for}~i+j~{\rm is~even}\\ 
f_{ij}R & {\rm for}~i+j~{\rm is~odd},
\end{array}\right. \nonumber
\end{eqnarray}
where $f_{i,j} \in \mathbb{R}$. According to Proposition~\ref{propo:1}, let $G=[G_{i,j}]=EF=I_{4N}$ ($1\leq i,j \leq 2N$), with 
\begin{eqnarray}
G_{i,j}=\left\{ \begin{array}{ll}
g_{ij}I_2 & {\rm for}~i+j~{\rm is~even}\\ 
g_{ij}R & {\rm for}~i+j~{\rm is~odd}.
\end{array}\right. \nonumber
\end{eqnarray}
Moreover, as  $G=I_{4N}$, we have $g_{i,j}=1$ for $i=j$ and $g_{i,j}=0$ for $i \neq j$. Let
\begin{eqnarray}\tilde{E}= \left[\begin{array}{cccc:ccc}
e_{1,1}&e_{1,2}&\cdots & e_{1, N}&e_{2N,N}   &\cdots     & e_{2N, 1}\\
e_{2,1}&e_{2,2}&\cdots & e_{2, N}&e_{2N-1,N}&\cdots & e_{2N-1,1}\\
\vdots & \vdots &\ddots &  \vdots & \vdots &\ddots & \vdots \\
e_{2N,1}&e_{2N,2}&\cdots & e_{2N, N}&e_{1,N}&\cdots & e_{1,1}
\end{array}\right]. \label{eq:tildeE} \nonumber
\end{eqnarray}
According to Definition~\ref{def:L2-matrix}, we only need to prove that
\begin{eqnarray}
\left(I_{N}  \otimes  \tilde{E} \right)
\left[ \begin{array}{c}
f_{1,1}\\
\vdots\\
f_{2N,1}\\
f_{1,2}\\
\vdots\\
f_{2N,2}\\
\vdots\\
f_{1,N}\\
\vdots\\
f_{2N,N}
\end{array}\right]
=
\left[ \begin{array}{c}
g_{1,1}\\
\vdots\\
g_{2N,1}\\
g_{1,2}\\
\vdots\\
g_{2N,2}\\
\vdots\\
g_{1,N}\\
\vdots\\
g_{2N,N}
\end{array}\right]
\end{eqnarray}
has a unique solution, which requires $\det\left(I_{N}  \otimes  \tilde{E} \right) \neq 0$, that is $\det(\tilde{E}_{i,j})\neq 0$,  based on the Cramer's rule. 
Define a permutation matrix 
\begin{eqnarray}
U=\left[\begin{array}{c}
I_{2N}  \otimes  \left[\begin{array}{cc} 1&0 \end{array}\right]\\
I_{2N}  \otimes  \left[\begin{array}{cc} 0&1 \end{array}\right]\\
 \end{array}\right], \label{eq:U}
\end{eqnarray}
then
\begin{eqnarray}
UEU^{T}= \left[\begin{array}{cc}
\tilde{E}&\\
&\hat{E}
\end{array}\right],  \label{eq:decomposition}
\end{eqnarray}
where 
\begin{eqnarray}
\hat{e}_{i,j} =\left\{ \begin{array}{ll}
\tilde{e}_{i,j} & {\rm for}~i+j~{\rm is~even}\\ 
-\tilde{e}_{i,j} & {\rm for}~i+j~{\rm is~odd},
\end{array}\right. 
\end{eqnarray}
and $\hat{e}_{i,j}$ and $\tilde{e}_{i,j}$ denote the elements at the $i_{th}$ row  and the $j_{th}$ column  of $ \hat{E}$ and $\tilde{E}$, respectively.
Since $\det(E) \neq 0$ and $\det(U)\neq 0$, we have $\det(\tilde{E})\neq 0$.  Therefore, $EF=I_{4N}$. Furthermore, we have $EFE=I_{4N}E$. Therefore $FE=E^{-1}EFE=E^{-1}I_{4N}E=I_{4N}$. The proof is completed.
\end{proof}

The following theorem gives values of the phase shifts $\theta_a$ and $\theta_b$ at which the two-mode squeezing of the system is optimal, as well as the formulae of the optimal two-mode squeezing spectra.
\begin{theorem} \label{theorem}
Let 
\begin{eqnarray}
&&\Upsilon=\frac{h_1^N}{\left(h_1h_2m_{N-1}+n_{N-1}-h_2^2n_{N-1} \right)^2\prod\limits_{k=0}^{N-2}n_k}\nonumber\\
&&\times \left(h_1h_2^2m_{N-1}-h_1^3 m_{N-1}+h_2n_{N-1}-h_2^3n_{N-1}+h_1^2h_2n_{N-1} \right). \nonumber
\end{eqnarray}
In the lossless case, the two-mode squeezing of a stable $N$-NOPA CFB system ($N \geq 2$) as shown in Fig.~\ref{fig:N-NOPA-cfb} is optimal when $\theta_a$ and $\theta_b$ ($-\pi<\theta_a,\theta_b \leq \pi$) satisfy
\begin{eqnarray}
|\theta_a+\theta_b|=\pi,~ {\rm for}~ \Upsilon >0; \nonumber\\
\theta_a+\theta_b=0 ~ {\rm or}~ \theta_a=\theta_b=\pi ~ {\rm for}~ \Upsilon <0; \nonumber\\
-\pi<\theta_a,\theta_b \leq \pi, ~ {\rm for}~ \Upsilon =0. \nonumber
\end{eqnarray}
The optimal two-mode squeezing $V_\pm$ in the infinite bandwidth limit are
\begin{eqnarray}
V_\pm = \left\lbrace  \begin{array}{c}
2\left(-h_2  +\frac{h_1^N- \left(h_1^2h_2n_{N-1}-h_1^3m_{N-1}\right)\prod\limits_{k=0}^{N-2}n_k}{\left(h_1h_2m_{N-1}+n_{N-1}-h_2^2n_{N-1} \right)\prod\limits_{k=0}^{N-2}n_k}\right)^2  ~ {\rm for}~ \Upsilon \geq 0;\\
2\left(h_2  +\frac{h_1^N+ \left(h_1^2h_2n_{N-1}-h_1^3m_{N-1}\right)\prod\limits_{k=0}^{N-2}n_k}{\left(h_1h_2m_{N-1}+n_{N-1}-h_2^2n_{N-1} \right)\prod\limits_{k=0}^{N-2}n_k}\right)^2~ {\rm for}~ \Upsilon < 0.
\end{array} \right. 
 \label{eq:N-cfb-Vpm} \nonumber
\end{eqnarray}
where
$n_0=1$, $n_1=1$, $n_{k+1}=1-h_2^2+\frac{h_1h_2m_{k}}{n_{k}}$, $m_1=0$, $m_{k+1}=-h_1h_2+\frac{h_1^2m_{k}}{n_{k}}$,  $1\leq k \leq N-2$, $h_1=\frac{(xy)^2+1}{(xy)^2-1}$ and $h_2=\frac{2xy}{(xy)^2-1}$. 
\end{theorem}
\begin{proof}
See the Appendix.
\end{proof}

Theorem~\ref{theorem} shows that the optimal two-mode squeezing of the coherent feedback system does not consistently 
exist in a certain pair of quadrature. That is, the pair of quadratures that has the best two-mode squeezing depends on the number of NOPAs.
Moreover, the optimized two-mode squeezing $V_+$ and $V_-$ are identical in the same pair of quadratures for different values of $N$. 

It is easy to check that we obtain the same results of the optimal two-mode squeezing and the corresponding values of $\theta_a$ and $\theta_b$ as given in Section~4.1.1 of \cite{SN2015qic}, when $N=2,3,4,5,6$.
In addition, Fig.~\ref{fig:compare} compares the optimal two-mode squeezing of $N$-NOPA CFB systems with $2 \leq N \leq 10$, under the same total pump power in the absence of losses and delays. 
We denote the parameters $x$ and $y$ of an $N$-NOPA CFB system as $x_N$ and $y_N$. We set $y_N=1$ for all $2 \leq N \leq 10$, $x_{10}=0.078$ that is close to the value of stability threshold (see \cite{SN2015qic}), and $x_i=(\sqrt{10/i})x_{10}$ for $2\leq i \leq 9$. 
The squeezing in the plot is presented in dB unit, that is, $V_\pm ({\rm dB})=10{\rm log}_{10}V_{\pm} $. It is shown that the system with more NOPAs generates better entanglement, under the same total pump power.

\begin{figure}[htbp]
\begin{center}
\includegraphics[scale=0.6]{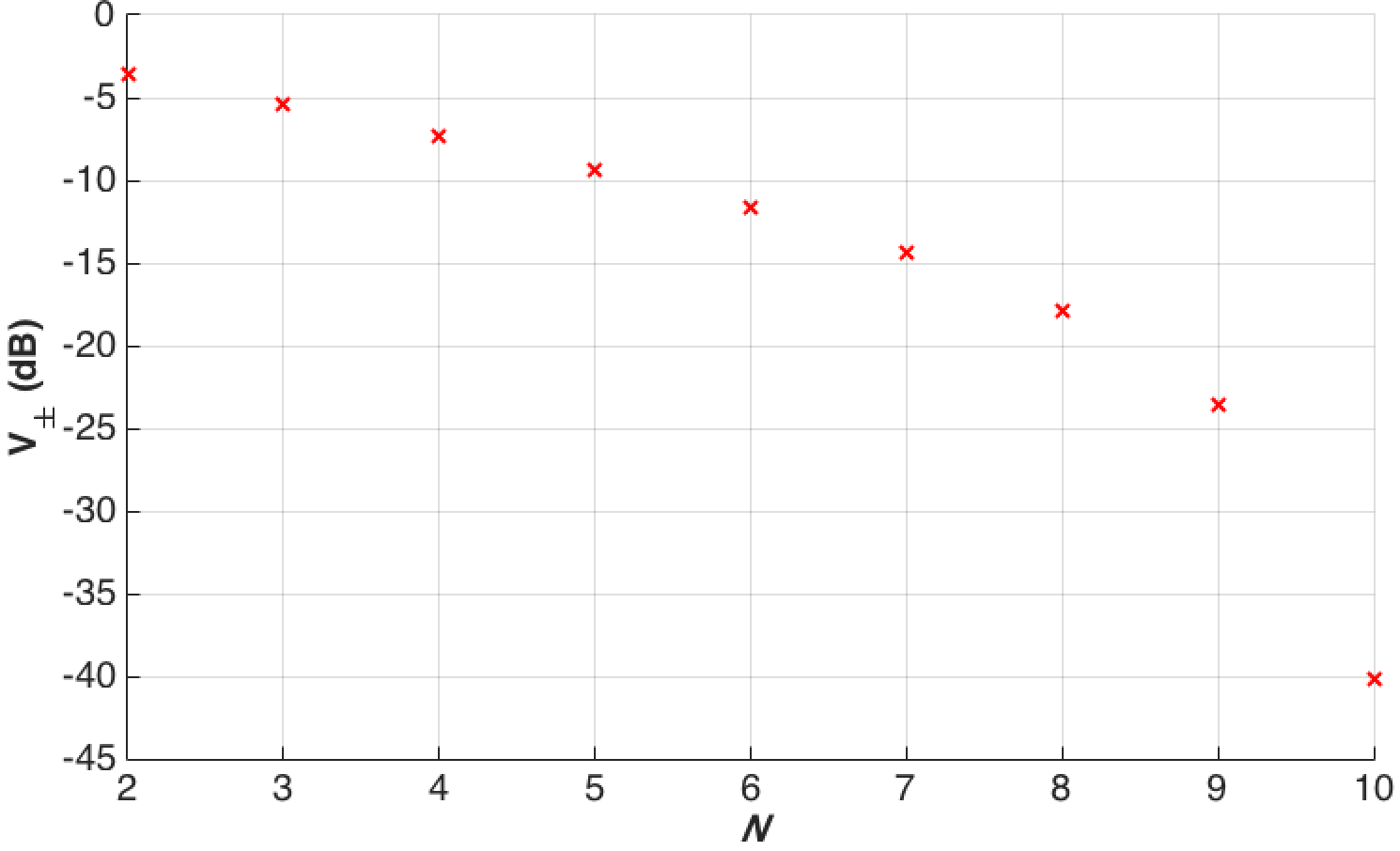}
\caption{Values of $V_{\pm} (\rm{dB})$ of $N$-NOPA systems ($2 \leq N \leq 10$) in the absence of losses and delays, under the same total pump power, with $x_{10}=0.13$, $x_i=(\sqrt{10/i})x_{10}$ ($i=\lbrace2,3,4,5,6,7,8,9\rbrace$) and $y=1$.}\label{fig:compare}
\end{center}
\end{figure}

\section{Conclusion}
\label{sec:conclusion}
This paper has derived formulae for EPR entanglement of a lossless $N$-NOPA coherent feedback system with two adjustable phase shifters $\theta_a$ and $\theta_b$ placed at the two outputs of the system.  We describe the coherent feedback system by a more general $N$-NOPA system containing a static passive linear network in the infinite bandwidth limit, which gives a good approximation to the EPR entanglement in the low frequency region.  Here, the passive network represents the coherent feedback connections. 
When losses are neglected, we give the values of $\theta_a$ and $\theta_b$ where the best two-mode squeezing is achieved, and derive formulae for the optimal two-mode squeezing.

As this work is limited to the case where losses are neglected, future study is required to formulate two-mode squeezing of the system under effects of losses. 
Moreover, quantitative analysis of EPR entanglement of the system  as $N$ approches $\infty$  may be of interest.
Furthermore, the $N$-NOPA general system in the infinite bandwidth limit can be further exploited to analyze other $N$-NOPA systems where the NOPAs are interconnected in a certain way or connected by a passive linear optical network consisting of passive optical devices.

\section*{APPENDIX} 
The below is the proof of Theorem~\ref{theorem}.

\begin{proof}
\textbf{Part 1:}  Let $\tilde{S}_N= \tilde{S}_N^{cfb}$.  In the ideal case ($\kappa=0$), from (\ref{eq:h_coeffi_static}), (\ref{eq:ScfbN}) and (\ref{eq:H}), we obtain the transfer matrix of the $N$-NOPA CFB system as
\begin{eqnarray}
&&H_{N}=\left(S_{N,11}+S_{N,12}\left(I_N \otimes W_{12}\right)P_{N} S_{N,21}\right)\left[\begin{array}{cc}I_{4}&O_{4\times 4N}\end{array}\right]   \nonumber\\ 
&&=\left[\begin{array}{ccccc}\left(M_1\otimes I_2 \right)&O_4&\cdots & O_4 &\left(M_2\otimes I_2 \right)\end{array}\right] \left(I_N \otimes W_{12}\right)P_{N}
\left[\begin{array}{c}\left(M_2\otimes I_2 \right)\\O_4\\ \vdots \\ O_4 \\ \left(M_1 \otimes I_2 \right)\end{array}\right] 
\left[\begin{array}{cc}I_{4}&O_{4\times 4N}\end{array}\right]\nonumber \\
&&=\left[\begin{array}{ccccc}\left[\begin{array}{cc}O_2&O_2\\ h_2 R& h_1I_2\end{array}\right]&O_4&\cdots & O_4 &\left[\begin{array}{cc}h_1 I_2&h_2R\\ O_2&O_2\end{array}\right]\end{array}\right]  P_{N}\left[\begin{array}{c}\left(M_2\otimes I_2 \right)\\O_4\\ \vdots \\ O_4 \\ \left(M_1 \otimes I_2 \right)\end{array}\right]\left[\begin{array}{cc}I_{4}&O_{4\times 4N}\end{array}\right]. \nonumber
\end{eqnarray}
Now let us investigate the matrix $P_{N}=\left(I_{4N}-S_{N,22}\left(I_N \otimes W_{12} \right) \right)^{-1}$, where $W_{12}$ is as given in (\ref{eq:H}).

It is easy to observe that $I_{4N}-S_{N,22}\left(I_N \otimes W_{12} \right)$
is an $L^\curvearrowright _{2\times 2}$-matrix, according to Definition~\ref{def:L2-matrix}. Thus, the matrix $P_N$ is an $L^\curvearrowright _{2\times 2}$-matrix, based on Proposition~\ref{propo:2}. Let $P_{N}=[P_{i,j}]$ where $P_{i,j} \in \mathbb{R}^{2 \times 2}$ and  let  $p_{i,j} \in \mathbb{R}$ be the element at the $i_{th}$ row and $j_{th}$ column of $P_N$. 
We have
\begin{eqnarray}
P_{i,j}&=&\left[\begin{array}{cc}
p_{2i-1,2j-1}& p_{2i-1,2j}\\
p_{2i,2j-1}& p_{2i,2j}\\
\end{array}\right].
\end{eqnarray}
Following Definition~\ref{def:L2-matrix}, for $1 \leq i \leq 2N$ and $N+1 \leq i \leq 2N$
\begin{eqnarray}
&&p_{2i-1,2j}=p_{2i,2j-1}=p_{4N+1-2i,4N+2-2j}=p_{4N+2-2i,4N+1-2j}=0, \nonumber\\
&&p_{2i-1,2j-1}=p_{2i,2j}=p_{4N+1-2i,4N+1-2j}=p_{4N+2-2i,4N+2-2j}, ~{\rm for~} i+j ~{\rm is~ even }, \nonumber\\
&&p_{2i-1,2j-1}=-p_{2i,2j}=p_{4N+1-2i,4N+1-2j}=-p_{4N+2-2i,4N+2-2j}, ~{\rm for~} i+j ~{\rm is~ odd }.  \label{eq:pij}
\end{eqnarray}
By using (\ref{eq:pij}), we obtain a transfer matrix of the form
\begin{eqnarray}
H_{N}=\left[\begin{array}{cccc}
u&0&v&0\\
0&u&0&-v\\
v&0&u&0\\
0&-v&0&u
\end{array}\right]\left[\begin{array}{cc}I_{4}&O_{4\times 4N}\end{array}\right], \label{eq:Hcfb}
\end{eqnarray}
where $u=h_1 p_{4N-3,1}+h_2 p_{4N-1,1}$ and $v=h_1 p_{3,1}+h_2 p_{1,1}$.
Let $Q_N=\left(P_N\right)^{-1}= I_{4N}-S_{N,22}\left(I_N \otimes W_{12} \right)$ and $q_{i,j}$ be the matrix element at the $i$-th row and $j$-th column.  As the $1$-st row of $Q_N$ has only one non-zero element $q_{1,1}$,  thus $p_{1,1}=1$. Moreover, as the second last row of a matrix formed by removing the $1$-st row and $(4N-1)$-th column of $Q_N$ is a zero vector, $p_{4N-1,1}=0$.
Therefore, we obtain $u=h_1 p_{4N-3,1}$ and $v=h_1 p_{3,1}+h_2$.

Define $\xi_{out,a}^Q+\xi_{out,b}^Q=H^Q\xi$ and $\xi_{out,a}^P-\xi_{out,b}^P=H^P\xi$, where $\xi_{out,a}=\xi_{out,1}e^{\imath \theta_a}$ and $\xi_{out,b}=\xi_{out,2}e^{\imath \theta_b}$. We have 
\begin{eqnarray}
&&H^Q=\left[\begin{array}{cccc}
\cos \theta_a &-\sin \theta_a & \cos \theta_b & -\sin \theta_b
\end{array}\right]H_{N}  \nonumber\\
&&H^P=\left[\begin{array}{cccc}
\sin \theta_a &\cos \theta_a & -\sin \theta_b & -\cos \theta_b
\end{array}\right]H_{N}. \nonumber
\end{eqnarray}
Following (\ref{eq:V_+}) and (\ref{eq:V_-}), we obtain that 
\begin{eqnarray}
V_\pm(\theta_a, \theta_b)=2\left(u^2+v^2+2uv\cos(\theta_a+\theta_b) \right). \label{eq:Vuv}
\end{eqnarray}
Let $\Upsilon=uv$. It follows that the two-mode squeezing spectra are optimal when $\theta_a$ and $\theta_b$ ($-\pi<\theta_a,\theta_b \leq \pi$) satisfy
\begin{eqnarray}
|\theta_a+\theta_b|=\pi,~ {\rm for}~ \Upsilon >0; \nonumber\\
\theta_a+\theta_b=0 ~ {\rm or}~ \theta_a=\theta_b=\pi ~ {\rm for}~ \Upsilon <0; \nonumber\\
-\pi<\theta_a,\theta_b \leq \pi, ~ {\rm for}~ \Upsilon =0. \nonumber
\end{eqnarray}
The optimal two-mode squeezing $V_\pm$ in the infinite bandwidth limit are
\begin{eqnarray}
V_\pm &=& 2\left(u^2+v^2-2|uv|\right) 
=\left\lbrace \begin{array}{c}
2\left(u-v\right)^2 ~ {\rm for}~ \Upsilon \geq 0; \nonumber\\
2\left(u+v\right)^2 ~ {\rm for}~ \Upsilon <0. \nonumber\\
\end{array}   \right. \label{eq:Vuv_optimized}
\end{eqnarray}

\textbf{Part 2:}
To calculate $u$ and $v$, now let us compute $p_{3,1}$ and $p_{4N-3,1}$.
Let $R_{1,N}$ be a matrix formed by removing the first row and the third 
\begin{eqnarray}
T_{1,N}=\left[\begin{array}{c|c}
1&\begin{array}{ccc}0&\cdots &0\end{array} \\\hline
\begin{array}{c}0\\\vdots \\0\end{array}&R_{1,N}
\end{array}\right] . \label{eq:T1N}
\end{eqnarray}
Thus, $p_{3,1}=\frac{\det(R_{1,N})}{\det(Q_N)}=\frac{\det(T_{1,N})}{\det(Q_N)}$.
To derive $\det(T_{1,N})$, we first introduce matrices
\begin{eqnarray}
&&T_{1,B}=\left[\begin{array}{cc}
O_2&O_2\\
-h_2R&-h_1I_2
\end{array}\right], \nonumber\\
&&T_{1,C}=\left[\begin{array}{cc}
-h_1I_2&-h_2R\\
O_2&O_2
\end{array}\right],  \nonumber\\
&&T_{1,D}(m_k,n_k)=\left[\begin{array}{cc}
I_2&O_2\\
m_kR&n_kI_2
\end{array}\right] \nonumber\\
&&T_{1}(m_k,n_k)=\left[\begin{array}{cccccccc}
1&0&0&0&0&0&0&0\\
0&0&1&0&0&0&0&0\\
0&0&0&0&-h_2&0&-h_1&0\\
0&0&0&1&0&h_2&0&-h_1\\
0&-h_1&0&0&1&0&0&0\\
0&0&-h_1&h_2&0&1&0&0\\
0&0&0&0&m_k&0&n_k&0\\
0&0&0&0&0&-m_k&0&n_k
\end{array}\right],\label{eq:T1submatrices}
\end{eqnarray}
where $k=1, 2, \cdots$, $m_1=0$ and $n_1=1$.
Moreover, we obtain that
\begin{eqnarray}
&&\det \left(T_{1D}(m_k,	n_k)\right)=n_k^2 \nonumber\\
&&T_{1D}(m_1,n_1)-T_{1B} T_{1D}(m_k,n_k)^{-1} T_{1C}=T_{1D}(m_{k+1},n_{k+1})\nonumber\\
&&\det \left(T_1(m_k,n_k)\right) =-h_1(h_1m_k-h_2n_k)(h_1h_2m_k-h_2^2n_k), \label{eq:T1formula}
\end{eqnarray}
where $m_{k+1}=-h_1h_2+\frac{h_1^2m_{k}}{n_{k}}$ and $n_{k+1}=1-h_2^2+\frac{h_1h_2m_{k}}{n_{k}}$.

For $N=2$, $T_{1,2}=T_1(m_1,n_1)$.
Thus, $\det(T_{1,2})=-h_1(h_1m_1-h_2n_1)(h_1h_2m_1+n_1-h_2^2n_1)$.

For $N\geq 3$,
\begin{eqnarray}
T_{1,N}=\left[\begin{array}{c|c}
T_{1,N-1}& \begin{array}{c}O_{4(N-2)\times 4}\\T_{1,B} \end{array} \\ \hline
\begin{array}{cccc}O_{4\times 4(N-2) }&T_{1,C} \end{array}&T_{1D}(m_1,n_1)
\end{array}\right] . \label{eq:T1N-2}
\end{eqnarray}
Exploiting  (\ref{eq:T1formula}) and the fact that $\det \left[\begin{array}{cc}
A&B\\C&D \end{array}\right]= \det(D)\det(A-BD^{-1}C)$ when $D$ is invertible, we have
\begin{eqnarray}
\det(T_{1,N})&&=\det(T_{1D}(m_1,n_1)\det\left( \left[\begin{array}{c|c}
T_{1,N-2}& \begin{array}{c}O_{4(N-3)\times 4}\\T_{1,B} \end{array} \\ \hline
\begin{array}{cccc}O_{4\times 4(N-3) }&T_{1,C} \end{array}&T_{1D}(m_1,n_1)
\end{array}\right]
   \right)\nonumber\\
&&=\prod\limits_{k=1} ^{N-3}\det\left(T_{1D}(m_k,n_k)\right)
\det\left( \left[\begin{array}{c|c}
T_{1,2}& \begin{array}{c}O_{4}\\T_{1,B} \end{array} \\ \hline
\begin{array}{cccc}O_4&T_{1,C} \end{array}&T_{1D}(m_{N-2},n_{N-2})
\end{array}\right]
   \right)\nonumber\\
&&=\prod\limits_{k=1} ^{N-2}\det\left(T_{1D}(m_k,n_k)\right)\det\left(T_1(m_{N-1},n_{N-1}\right)\nonumber\\
&&=\prod\limits_{k=1} ^{N-2} n_k^2(-h_1(h_1m_{N-1}-h_2n_{N-1}) (h_1h_2m_{N-1}+n_{N-1}-h_2^2n_{N-1})).\label{eq:detT1N}
\end{eqnarray}

Let $R_{2,N}$ be a matrix formed by removing the first row and the $(4N-3)$-th column of the matrix $Q_N$. Define 
\begin{eqnarray}
T_{2,N}=\left[\begin{array}{c|c}
R_{2,N}&\begin{array}{c}0\\\vdots \\0\end{array} \\\hline
\begin{array}{ccc}0 &\cdots &0\end{array}&1
\end{array}\right] . \label{eq:T2N}
\end{eqnarray}
Thus, $p_{4N-3,1}=\frac{\det(R_{2,N})}{\det(Q_N)}=\frac{\det(T_{2,N})}{\det(Q_N)}$.
To derive $\det(T_{2,N})$, we firstly introduce matrices
\begin{eqnarray}
&&T_{2A}(m_k,n_k)=\left[\begin{array}{cccc}
0&n&0&-m\\
0&0&1&0\\
0&0&0&1\\
-h_1&0&-h_2&0
\end{array}\right], \nonumber\\
&&T_{2B}=\left[\begin{array}{cccc}
0&0&0&0\\
-h_2&0&-h_1&0\\
0&h_2&0&-h_1\\
1&0&0&0
\end{array}\right],  \nonumber\\
&&T_{2C}=\left[\begin{array}{cccc}
0&-h_1&0&h_2\\
0&0&0&0\\
0&0&0&0\\
0&0&0&0\\
\end{array}\right],  \nonumber\\
&&T_{2}(m_k,n_k)=\left[\begin{array}{cccccccc}
0&n_k&0&-m_k&0&0&0&0\\
0&0&1&0&0&-h_1&0&0\\
0&0&0&1&h_2&0&-h_1&0\\
-h_1&0&-h_2&0&0&0&0&0\\
0&-h_1&0&h_2&1&0&0&0\\
0&0&0&0&0&1&0&0\\
0&0&0&0&0&0&1&0\\
0&0&0&0&0&0&0&1
\end{array}\right],\label{eq:T2submatrices}
\end{eqnarray}
where $k=1, 2, \cdots$, $m_1=0$ and $n_1=1$.
We have
\begin{eqnarray}
&&\det (T_{2A}(m_k,n_k))=h_1n_k \nonumber\\
&&T_{2A}(m_1,n_1)-T_{2C} T_{2A}(m_k,n_k)^{-1} T_{2B}=T_{2A}(m_{k+1},n_{k+1})\nonumber\\
&&\det \left(T_{2}(m_k,n_k)\right) =h_1(h_1h_2m_k+n_k-h_2^2n_k), 
\label{eq:T2formula} 
\end{eqnarray}
where $m_{k+1}=-h_1h_2+\frac{h_1^2m_{k}}{n_{k}}$ and $n_{k+1}=1-h_2^2+\frac{h_1h_2m_{k}}{n_{k}}$.

For $N=2$,
$\det (T_{2,2})=\det (T_{2}(m_1,n_1))=h_1(h_1h_2m_1+n_1-h_2^2n_1)$.

For $N\geq 3$, 
\begin{eqnarray}
T_{2,N}=\left[\begin{array}{c|c}
T_{2A}(m_1,n_1)& \begin{array}{cccc} T_{2B} &O_{4\times 4(N-2)}  \end{array} \\ \hline
\begin{array}{c}T_{2C}\\O_{4(N-2)\times 4} \end{array}&T_{2,N-1}
\end{array}\right] .  \label{eq:T2N-2}
\end{eqnarray}
Exploiting  (\ref{eq:T2formula}) and the fact that $\det \left[\begin{array}{cc}A&B\\C&D \end{array}\right]= \det(A)\det(D-CA^{-1}B)$  when $A$ is invertible,  we obtain
\begin{eqnarray}
\det(T_{2,N})=h_1^{N-1}(h_1h_2m_{N-1}+n_{N-1}-h_2^2n_{N-1})  \prod\limits_{k=1} ^{N-2} n_k. \label{eq:detT2N}
\end{eqnarray}

At last, we derive $\det(Q_N)$. Let $T_{3,N}=Q_N$ and define
\begin{eqnarray}
&&T_{3D}(m_k,n_k)=T_{1D}(m_k,n_k) ,\nonumber\\
&&T_{3B}=T_{1B},\nonumber\\
&&T_{3C}=T_{1C},  \nonumber\\
&&T_{3}(m_k,n_k)=\left[\begin{array}{cc}
I_4&T_{3B}\\
T_{3C}&T_{3D}(m_k,n_k)
\end{array}\right],
\label{eq:T3submatrices}
\end{eqnarray}
where $k=1, 2, \cdots$, $m_1=0$ and $n_1=1$.
Moreover,
\begin{eqnarray}
&&\det (T_{3D}(m_k,n_k))=n_k^2 \nonumber\\
&&T_{3D}(m_1,n_1)-T_{3B} T_{3D}(m_k,n_k)^{-1} T_{3C}=T_{3D}(m_{k+1},n_{k+1})\nonumber\\
&&\det \left(T_3(m_k,n_k)\right) =(h_1h_2m_k+n_k-h_2^2n_k)^2, \label{eq:T3formula}
\end{eqnarray}
where $m_{k+1}=-h_1h_2+\frac{h_1^2m_{k}}{n_{k}}$ and $n_{k+1}=1-h_2^2+\frac{h_1h_2m_{k}}{n_{k}}$.

For $N=2$, $T_{3,2}=T_{3}(m_1,n_1)=(h_1h_2m_1+n_1-h_2^2n_1)^2$.

For $N \geq 3$,
\begin{eqnarray}
T_{3,N}=\left[\begin{array}{c|c}
T_{3,N-1}& \begin{array}{c}O_{4(N-2)\times 4}\\T_{3B} \end{array} \\ \hline
\begin{array}{cccc}O_{4\times 4(N-2) }&T_{3C} \end{array}&T_{3D}(m_1,n_1)
\end{array}\right] . \label{eq:T3N-2}
\end{eqnarray}
Exploiting  (\ref{eq:T3formula}) and the fact that $\det \left[\begin{array}{cc}
A&B\\C&D \end{array}\right]= \det(D)\det(A-BD^{-1}C)$ when $D$ is invertible, we have
\begin{eqnarray}
\det(T_{3,N})=(h_1h_2m_{N-1}+n_{N-1}-h_2^2n_{N-1})^2\prod\limits_{k=1} ^{N-2} n_k^2.  \label{eq:detT3N}
\end{eqnarray}

Therefore, we have
\begin{eqnarray}
u&=&h_1 p_{4N-3,1}=h_1\frac{\det\left(T_{2,N} \right)}{\det\left(T_{3,N} \right)} =\frac{h_1^N}{(h_1h_2m_{N-1}+n_{N-1}-h_2^2n_{N-1})\prod\limits_{k=0} ^{N-2} n_k}, \nonumber\\
v&=&h_1 p_{3,1}+h_2=h_1\frac{\det\left(T_{1,N} \right)}{\det\left(T_{3,N} \right)}+h_2 \nonumber\\
&=&h2+\frac{-h_1^2(h_1m_{N-1}-h_2n_{N-1})}{(h_1h_2m_{N-1}+n_{N-1}-h_2^2n_{N-1})} \nonumber\\
&=&\frac{h_1h_2^2m_{N-1}+h_2n_{N-1}-h_2^3n_{N-1}-h_1^3m_{N-1}+h_1^2h_2n_{N-1}}{(h_1h_2m_{N-1}+n_{N-1}-h_2^2n_{N-1})}. \nonumber
\end{eqnarray}
where
$N\geq 2$, $n_0=1$, $m_1=0$, $n_1=1$, $m_{k+1}=-h_1h_2+\frac{h_1^2m_{k}}{n_{k}}$ and $n_{k+1}=1-h_2^2+\frac{h_1h_2m_{k}}{n_{k}}$, $1 \leq k \leq N-2$.
The proof is completed.
\end{proof}

\bibliography{refs} 

\begin{thebibliography}{10}

\bibitem{Yonezawa2004}
H.~Yonezawa, T.~Aoki, and A.~Furusawa.
\newblock Demonstration of a quantum teleportation network for continuous
  variables.
\newblock {\em Nature}, 431:430, 2004.

\bibitem{Jouguet2013}
P.~Jouguet, S.~Kunz-Jacques, A.~Leverrier, P.~Grangier, and E.~Diamanti.
\newblock Experimental demonstration of long-distance continuous-variable
  quantum key distribution.
\newblock {\em Nature Photonics}, 7:378, 2013.

\bibitem{Briegel1998}
H.~J. Briegel, W.~D{\"u}r, J.~I. Cirac, and P.~Zoller.
\newblock Quantum repeaters: the role of imperfect local operations in quantum
  communication.
\newblock {\em Phys. Rev. Lett.}, 81:5932, 1998.

\bibitem{Horodecki2009}
R.~Horodecki, P.~Horodecki, M.~Horodecki, and K.~Horodecki.
\newblock Quantum entanglement.
\newblock {\em Rev. Mod. Phys.}, 81:865, 2009.

\bibitem{Bowen2004}
W.~P. Bowen, R.~Schnabel, P.~K. Lam, and T.~C. Ralph.
\newblock Experimental characterization of continuous-variable entanglement.
\newblock {\em Phys. Rev. A}, 69:012304, 2004.

\bibitem{Weedbrook2012}
C.~Weedbrook, S.~Pirandola, R.~Garcia-Patron, N.~J. Cerf, T.~C. Ralph, J.~H.
  Shapiro, and S.~Lloyd.
\newblock {Gaussian} quantum information.
\newblock {\em Rev. Mod. Phys}, 84:621, 2012.

\bibitem{Ou1992}
Z.~Y. Ou, S.~F. Pereira, and H.~J. Kimble.
\newblock Realization of the {Einstein-Podolsky-Rosen} paradox for continuous
  variables in nondegenerate parametric amplification.
\newblock {\em Appl. Phys. B}, 55:265, 1992.

\bibitem{Vitali2006}
D.~Vitali, G.~Morigi, and J.~Eschner.
\newblock Single cold atom as efficient stationary source of {EPR}-entangled
  light.
\newblock {\em Phys. Rev. A}, 74:053814, 2006.

\bibitem{He2007}
W.-P. He and F.-L. Li.
\newblock Generation of broadband entangled light through cascading
  nondegenerate optical parametric amplifiers.
\newblock {\em Phys. Rev. A}, 76:012328, 2007.

\bibitem{SN2015qip}
Z.~Shi and H.~I. Nurdin.
\newblock Coherent feedback enabled distributed generation of entanglement
  between propagating {Gaussian} fields.
\newblock {\em Quantum Information Processing}, 14:337, 2015.

\bibitem{SN2015qic}
Z.~Shi and H.~I. Nurdin.
\newblock Entanglement in a linear coherent feedback chain of nondegenerate
  optical parametric amplifiers.
\newblock {\em Quantum Information and Computation}, 15(13 \& 14), 2015.

\bibitem{Gough2009}
J.~E. Gough and S.~Wildfeuer.
\newblock Enhancement of field squeezing using coherent feedback.
\newblock {\em Phys. Rev. A}, 80:042107, 2009.

\bibitem{SN2015acc}
Z.~Shi and H.~I. Nurdin.
\newblock Optimization of distributed entanglement generated between two
  {Gaussian} fields by the modified steepest descent method.
\newblock In {\em Proceedings of the 2015 American Control Conference (Chicago,
  US, July 1-3, 2015)}.

\bibitem{bGardiner2004}
C.~W. Gardiner and P.~Zoller.
\newblock {\em Quantum Noise}.
\newblock Springer-Verlag Berlin Heidelberg, 3rd edition, 2004.

\bibitem{Nurdin2009}
H.~I. Nurdin, M.~R. James, and A.~C. Doherty.
\newblock Network synthesis of linear dynamical quantum stochastic systems.
\newblock {\em SIAM J. Control Optim.}, 48(4):2686, 2009.

\bibitem{inbBelavkin2008}
V.~P. Belavkin and S.~Edwards.
\newblock {\em Quantum Stochastics and Information - Statistics, Filtering and
  Control}, chapter Quantum filtering and optimal control, page 143.
\newblock World Scientific, University of Nottingham, UK, 2008.

\bibitem{bWiseman2010}
H.~M. Wiseman and G.~J. Milburn.
\newblock {\em Quantum Measurement and Control}.
\newblock Cambridge University Press, 2010.

\bibitem{Gough2003}
J.~E. Gough.
\newblock Quantum white noises and the master equation for {Gaussian} reference
  states.
\newblock {\em Russ. J. Math. Phys.}, 10(2):142, 2003.

\bibitem{Laurat2005}
J.~Laurat, G.~Keller, J.~A. Oliveira-Huguenin, C.~Fabre, T.~Coudreau,
  A.~Serafini, G.~Adesso, and F.~Illuminati.
\newblock Entanglement of two-mode {Gaussian} states: characterization and
  experimental production and manipulation.
\newblock {\em J. Opt. B: Quantum Semiclass. Opt.}, 7:S577, 2005.

\bibitem{Gough2010}
J.~E. Gough, M.~R. James, and H.~I. Nurdin.
\newblock Squeezing components in linear quantum feedback networks.
\newblock {\em Phys. Rev. A}, 81:023804, 2010.

\bibitem{Nurdin2012}
H.~I. Nurdin and N.~Yamamoto.
\newblock Distributed entanglement generation between continuous-mode
  {Gaussian} fields with measurement-feedback enhancement.
\newblock {\em Phys. Rev. A}, 86:022337, 2012.

\bibitem{bCarmichael2008}
H.~J. Carmichael.
\newblock {\em Statistical Methods in Quantum Optics 2}.
\newblock Springer-Verlag Berlin Heidelberg, 2008.

\bibitem{bBachor2009}
H-A. Bachor and T.~C. Ralph.
\newblock {\em A Guide to Experiments in Quantum Optics}.
\newblock Wiley-VCH, second, revised and enlarged edition, 2009.

\bibitem{Iida2012}
S.~Iida, M.~Yukawa, H.~Yonezawa, N.~Yamamoto, and A.~Furusawa.
\newblock Experimental demonstration of coherent feedback control on optical
  field squeezing.
\newblock {\em IEEE Trans. Autom. Control}, 57(8):2045, 2012.

\bibitem{Gough2008}
J.~E. Gough, R.~Gohm, and M.~Yanagisawa.
\newblock Linear quantum feedback networks.
\newblock {\em Phys. Rev. A}, 78:062104, 2008.

\end{thebibliography}
\bibliographystyle{unsrt}
\end{document}